\documentclass[12pt]{article}
\usepackage{slashbox}
\usepackage{multirow}
\usepackage[T1]{fontenc}
\usepackage[cp1250]{inputenc}

\usepackage{graphicx}
\usepackage{color}
\usepackage{enumerate}
\usepackage{rotating}
\usepackage{tikz}
\usepackage{amsmath, amsthm, amsfonts, amssymb}

\newtheorem{theorem}{Theorem}[section]

\newtheorem{corollary}[theorem]{Corollary}

\newtheorem{proposition}[theorem]{Proposition}

\newtheorem{lemma}[theorem]{Lemma}

\begin{document}
\markboth{ }{}
\title{\bf Equitable Colorings of $l$-Corona Products \\of Cubic Graphs\footnote{Project has been 
partially supported by Narodowe Centrum Nauki under contract 
DEC-2011/02/A/ST6/00201}}
\date{}
\author{Hanna Furma\'nczyk\footnote{Institute of Informatics,\ University of Gda\'nsk,\ Wita Stwosza 57, \ 80-308 Gda\'nsk, \ Poland. \ e-mail: hanna@inf.ug.edu.pl},  
\ Marek Kubale \footnote{Department of Algorithms and System Modelling,\ Gda\'nsk University of Technology,\ Narutowicza 11/12, \ 80-233 Gda\'nsk, \ Poland. \ e-mail: kubale@eti.pg.gda.pl},
}

\markboth{H. Furma\'nczyk, M. Kubale}{Equitable Colorings of $l$-Corona Products of Cubic Graphs}

\maketitle

\begin{abstract}
A graph $G$ is equitably $k$-colorable if its vertices can be partitioned into $k$ independent sets in such a way that the number of vertices in any two sets differ by at most one. The 
smallest integer $k$ for which such a coloring exists is known as the \emph{equitable chromatic number} of $G$ and it is denoted by $\chi_{=}(G)$. 

In this paper the problem of determinig the value of equitable chromatic number for multicoronas of cubic graphs $G \circ^l H$ is studied. The problem of ordinary coloring of multicoronas of 
cubic graphs is solvable in polynomial time. 
The complexity of equitable coloring problem is an open question for these graphs. We provide some polynomially solvable cases of cubical multicoronas and give simple linear time 
algorithms for equitable coloring of such graphs which use at most $\chi_=(G \circ ^l H) + 1$ colors in the remaining cases. 
\end{abstract}


\section{Introduction}\label{intro}
All graphs considered in this paper are connected, finite and simple, i.e. undirected, loopless and without multiple edges.

The paper concerns one of popular graph coloring models, namely equitable coloring. If the set of vertices of a graph $G$ can be partitioned into $k$ (possibly empty) classes 
$V_1, V_2, \ldots,V_k$ such that each $V_i$ is an independent set and the condition $||V_i|-|V_j||\leq 1$ holds for every pair ($i, j$), then $G$ is said to be {\it equitably k-colorable}. 
The smallest integer $k$ for which $G$ is equitably $k$-colorable is known as the {\it equitable chromatic number} of $G$ and denoted by $\chi_{=}(G)$ \cite{meyer}. When the condition 
$||V_i|-|V_j||=0$ holds for every pair ($i, j$), graph $G$ is said to be \emph{strong equitably $k$-colorable}. Given a $k$-coloring of $G$, a vertex with color $i$ is called an 
$i$-\emph{vertex}.

It is interesting to note that if a graph $G$ is equitably $k$-colorable,
it does not imply that it is equitably $(k + 1)$-colorable. A counterexample is the complete bipartite graph (also cubic graph) $K_{3,3}$, which can be equitably colored with two colors, but not with three. The smallest integer $k$, for which $G$ is equitably $k'$-colorable for all $k' \geq k$, is called the \emph{equitable chromatic threshold} of $G$ and denoted by $\chi^*_=(G)$.

We use also the concept of semi-equitable coloring. Graph $G$ has a \emph{semi-equitable $k$-coloring},
if there exists a partition of its vertices into independent sets $V_1, \ldots , V_k \subset V$ such that
one of these subsets, say $V_i$, is of size $\not\in \{\lfloor n/k \rfloor, \lceil n/k \rceil \}$, and the remaining subgraph $G - V_i$ is equitably $(k - 1)$-colorable. Note that not all
graphs have such a coloring, for example $K_4$ does not have.
In the following we will say that graph $G$ has $(V_1, \ldots, V_k)$-coloring to express explicitly a partition of $V$ into $k$ independent sets.
If, however, only cardinalities of color classes are important, we will use the notation $[|V_1|, \ldots, |V_k|]$.

The model of equitable graph coloring has many practical applications. Every time when we have to divide a system with binary conflict relations into equal or almost equal conflict-free 
subsystems we can model this situation by means of equitable graph coloring. Furma\'nczyk \cite{furm} mentions a specific application of this type of scheduling problem, namely, assigning 
university courses to time slots in a way that avoids scheduling incompatible courses at the same time and spreads the courses evenly among the available time slots. Also, the application of 
equitable coloring in scheduling of jobs on uniform machines was considered in \cite{3masz} and \cite{4masz}.

In this paper we consider the problem of equitable vertex-coloring for one of known graph products, namely for corona products of cubic graphs. Graph products are interesting and 
useful in many situations. The complexity of many problems, also equitable coloring, that deal with very large and complicated graphs is reduced greatly if one is able to fully characterize 
the properties of less complicated prime factors. Moreover, corona graphs lie often close to the boundary between easy and hard coloring problems \cite{harder}. More formally, 
the \emph{corona} of two graphs, $n_G$-vertex graph $G$ and 
$n_H$-vertex graph $H$, is a graph $G \circ H$ formed from one copy of $G$,  called the \emph{center graph}, and $n_G$ copies of $H$, named the 
\emph{outer graph}, where the $i$-th vertex of $G$ is adjacent to every vertex in the $i$-th copy of $H$. Such type of graph product was introduced by Frucht and Harary \cite{frucht}. In this 
paper we extend this concept to $l$-corona products as follows. For any integer $l \geq 2$, we define the graph $G \circ ^l H$ as 
$G \circ ^l H = (G \circ ^{l-1} H ) \circ H$, where $G \circ ^1 H =G \circ H$. Graph $G \circ ^l H$ is also named as \emph{$l$-corona product} of $G$ and $H$. 

The problem of equitable 
coloring of corona products of cubic graphs was considered in \cite{harder}. The authors showed that although the problem of ordinary coloring of coronas of cubic graphs is solvable in polynomial time, the problem of equitable coloring becomes 
NP-hard for such graphs. Moreover, they provided polynomially solvable instances of cubical coronas in some cases and 1-absolute approximation algorithms in the remaining cases.
In this paper we extend the results from \cite{harder} for cubical multicoronas. 
%Note also that the concept of equitable coloring of multicoronas was first considred in \cite{multivahan}. 

Now, let us recall some facts concerning cubic graphs. In 1994, Chen et al. \cite{clcub} proved that for every connected cubic graph, the chromatic number of which is 3, 
the equitable chromatic number of it is also equal to 3. Moreover, since connected cubic graph $G$, for which $\chi(G)=2$ is a bipartite graph $G(A,B)$ such that $|A|=|B|$, we have:
$$\chi(G)=\chi_=(G)$$ 
and due to Brooks Theorem \cite{brooks}:
$$2 \leq \chi_=(G) \leq 4,$$
for any cubic graph $G$.

Let
\begin{itemize}
\item $Q_2$ denote the class of $2$-chromatic cubic graphs,
\item $Q_3$ denote the class of $3$-chromatic cubic graphs,
\item $Q_4$ denote the class of $4$-chromatic cubic graphs.

Clearly, $Q_4=\{K_4\}$.
\end{itemize}

Next, let $Q_2(t) \subset Q_2$ ($Q_3(t) \subset Q_3$) denote the class of 2-chromatic (3-chromatic) cubic graphs with partition sets of cardinality $t$, and let $Q_3(u,v,w) \subset Q_3$ denote the class of 
3-chromatic graphs with color classes of cardinalities $u$, $v$ and $w$, respectively, where $u \geq v \geq w \geq u-1$.

Hajnal and Szemeredi \cite{hfs:haj} proved
\begin{theorem}
If $G$ is a graph satisfying $\Delta(G) \leq k$ then $G$ has an equitable $(k + 1)$-coloring. 
\end{theorem}

This theorem implies that every subcubic graph has an equitable $k$-coloring for every $k \geq 4$. In other words,
\begin{equation}
\chi_=^*(G) \leq 4.\label{prog}
\end{equation}

This result was extended in \cite{sharp} into a semi-equitable coloring of cubic graphs.
\begin{theorem}[\cite{sharp}]
Given an $n$-vertex subcubic graph $G$, $K_4 \neq G \neq K_{3,3}$, a constant $k \geq 4$, and an integer function $s=s(n)$. There is a semi-equitable $k$-coloring of $G$ of type 
$[s,\lceil \frac{n-s}{k-1} \rceil, \ldots,\lfloor \frac{n-s}{k-1}\rfloor]$, if $s \leq \lceil n/3 \rceil$.
\label{semi}
\end{theorem}

The remainder of the paper is organized as follows. In Section \ref{aux} we provide some auxilary tools while in Section \ref{main} we give our main results. Namely, we give in some cases 
polynomial algorithms for optimal equitable coloring of cubical coronas $G \circ^l H$, $l \geq 1$, while in the remaining 
cases we give sharp bounds on the equitable chromatic number of $l$-corona products of such
graphs.
Section \ref{sum} summarizes our results in a tabular form and remains as an open question the complexity status of equitable coloring of graphs under consideration.

\section{Some auxilaries} \label{aux}

In this section we prove lemmas, which are used very often in the further part of the paper.
 
\begin{lemma}
Let $G$, $H$ be cubic graphs and $l \geq 1$. If $G$ has a strong equitable $k$-coloring, then so does $G \circ^l H$ for any $k \geq 5$.
\label{lma}
\end{lemma}
\begin{proof}
First, notice that every cubic graph $H$ can be seen as a $(k-1)$-partite graph $H(X_1,X_2,\ldots, X_{k-1})$, $k \geq 5$ (due to inequality (\ref{prog})). Next, let 
$V(G)=V_1 \cup V_2 \cup \cdots \cup V_k$, where $V_1, \ldots, V_k$ are independent sets each of size $n_G/k$ (due to strong equitable $k$-coloring of $G$). 

Now, we determine an equitable $k$-coloring of $G \circ^1 H$, starting from the strong equitable $k$-coloring of $G$: $c:V(G) \rightarrow \{1,2, \ldots, k\}$. 
We extend it to the copies of $H$ in $G \circ H$ in the following way:
\begin{itemize}
\item color vertices of each copy of $H$ linked to an $i$-vertex of $G$ using color $(i+j) \bmod k$ for vertices in $X_j$ (we use color $k$ instead of 0), for $i=1,\ldots,k$ and $j=1,\ldots,k-1$.
\end{itemize}

Let us notice that this $k$-coloring of $G \circ^1 H$ is strong equitable. Indeed, every color is used $n_G/k+n_G/k(|X_1|+|X_2|+\cdots +|X_{k-1}|)=n_G(n_H+1)/k$ times. 
The thesis follows from  induction on $l$. 
\end{proof}

\begin{lemma}
Let $G$, $H$ be cubic graphs, $H \in Q_2 \cup Q_3$, and $l \geq 1$. If $G$ has a strong equitable $4$-coloring, then so does $G \circ^l H$.
\label{lma4}
\end{lemma}
\begin{proof}
Let $V(G)=V_1 \cup V_2 \cup V_3 \cup V_4$, where $V_1, \ldots, V_4$ are independent sets of size $n_G/4$ each.
Now, we determine an equitable $4$-coloring of $G \circ^1 H$, starting from the strong equitable $4$-coloring of $G$: $c:V(G) \rightarrow \{1,2, 3, 4\}$. 
We extend it to the copies of $H$ in $G \circ H$ in the following way:
\begin{description}
\item[\textnormal{\emph{Case} 1:}] $H \in Q_2$ 

Let $H=H(X_1,X_2)$. We color the vertices of each copy of $H$ linked to an $i$-vertex of $G$ using color $(i+j) \bmod 4$ for vertices in $X_j$ (we use color 4 instead of 0), for $i=1,\ldots,4$ 
and $j=1,2$.

\item[\textnormal{\emph{Case} 2:}] $H \in Q_3$ 

Let $H=H(X_1,X_2,X_3)$. We color vertices of each copy of $H$ linked to an $i$-vertex of $G$ using color $(i+j) \bmod 4$ for vertices in $X_j$ (we use color 4 instead of 0), for $i=1,\ldots,4$ 
and $j=1,2,3$.
\end{description}

Notice that the $4$-coloring of $G \circ^1 H$ is strong equitable. Indeed, every color is used $n_G(n_H+1)/4$ times. 
The thesis follows from induction on $l$. 
\end{proof}

Actually, we have proved 

\begin{corollary}
Let $G$ and $H \neq K_4$ be cubic graphs and $l \geq 1$. If $G$ has a strong equitable $k$-coloring, then 
$$\chi_=(G \circ ^l H) \leq k,$$
 for any $k \geq 4$.\label{lma_cor}
\end{corollary}

\section{Main results}\label{main}
\subsection{Case $H \in Q_2$}
In this subsection we obtain exact values of $\chi_=(G \circ ^l H)$, where $H \in Q_2$. We give also polynomial-time algorithms for the corresponding colorings.

First, let us recall a known result.
\begin{proposition}[\cite{harder}]
If $G$ is any cubic graph and $H \in Q_2$, then
$$\chi_=(G \circ H)=
\left \{
\begin{array}{ll}
3 & if \ G \neq K_{3,3} \ and \ 6|n_G, \\
4 & otherwise.
\end{array}
\right.
$$\label{3-4} 
\end{proposition}

\begin{theorem}
If $G$ and $H$ are cubic graphs, then
$\chi_=(G \circ ^l H) =3$ if and only if $G$ has a strong equitable $3$-coloring and $H \in Q_2$.\label{3col}
\end{theorem}
\begin{proof}
\noindent $(\Leftarrow)$ The truth for $l=1$ follows from Proposition \ref{3-4}. For $l >1$, since $n_G$
is divisible by 6 by the assumption, and $|V(G \circ ^ l H)|=n_G (n_H + 1)^l$, so the number of vertices in multicorona $G \circ ^l H$ is divisible by 6 for all $l \geq 1$. 
Hence, we get the thesis inductively.

\noindent $(\Rightarrow)$ Assume that $\chi_=(G \circ^l H) =3$. This implies:
\begin{itemize}
\item $H$ must be 2-chromatic,
\item $G$ must be 3-colorable (not necessarily equitably), i.e. $\chi(G) \leq 3$, which implies $G \in Q_2 \cup Q_3$.
\end{itemize}
Otherwise, we would have $\chi(G \circ^l H)  \geq 4$, which is a contradiction. 

We begin with $l=1$. Since $H\in Q_2$ is connected, its bipartition is determined. Let $H\in Q_2(t)$, $t\geq 3$. Observe that every 3-coloring of $G$ determines a 3-partition of 
$G \circ H$. Let us consider any 3-coloring of $G$ with color classes of cardinality $n_1, n_2$ and $n_3$, respectively, where $n_G=n_1+n_2+n_3$ and $n_1 \geq n_2 \geq n_3$. Then the cardinalities of 
color classes in the 3-coloring of $G \circ H$ form a sequence $(n_1^1,n_2^1,n_3^1)=(n_1+(n_2+n_3)t, n_2+ (n_1+n_3)t, n_3+(n_1+n_2)t)$. Such a 3-coloring of $G \circ H$ is equitable if and 
only if $n_1=n_2=n_3$. This means that $G$ must have a strong equitable 3-coloring.

For greater $l$ the cardinalities of color classes in the determined 3-coloring of $G \circ ^l H$, $(n_1^l, n_2^l, n_3^l)$, can be computed from the recursion:

$$
\left\{
\begin{array}{lcl}
n_1^0 &= &n_1,\\
n_2^0& = & n_2,\\
n_3^0 & = & n_3.
\end{array}\right.
$$
For $l \geq 1$:
$$
\left\{
\begin{array}{lcl}
n_1^l & = & n_1^{l-1} + (n_2^{l-1}+n_3^{l-1})t,\\
n_2^l & = & n_2^{l-1} + (n_1^{l-1}+n_3^{l-1})t,\\
n_3^l & = & n_3^{l-1} + (n_1^{l-1}+n_2^{l-1})t.
\end{array}\right.
$$

One can observe that in the determined 3-coloring of $G \circ ^l H$ the following statements are true:
\begin{itemize}
\item the color classes with the biggest difference between their cardinalities are classes of colors 1 and 3,
\item the order relation between cardinalities of color classes of colors 1 and 3 changes alternately, namely $n_1^l \geq n_3^l$ for odd $l$ while $n_1^l \leq n_3^l$ for even $l$,
\item the absolute value of the difference between cardinalities of color classes for colors 1 and 3 does not decrease as $l$ goes to infinity.
\end{itemize}

Due to the above, the 3-coloring of $G \circ^l H$ is equitable if only $n_1=n_2=n_3$, which completes the proof.
\end{proof}

\begin{theorem}
If $G$ is an arbitrary cubic graph and $H \in Q_2$, $l \geq 1$, then
$$\chi_=(G \circ^l H)\leq 4.$$
\label{H_2leq4}
\end{theorem}
\begin{proof} Let $H \in Q_2(y)$.

\begin{description}
\item[\textnormal{\emph{Case} 1:}] $G \in Q_2(x)$ 

We start from an equitable 4-coloring of graph $G$ such that each of colors 1 and 3 is used $\lceil x/2\rceil$ times, while each of colors 2 and 4 is used $\lfloor x/2\rfloor$ times. Next, we color  each copy of graph $H$ in $G \circ H$ with two colors in the following way:
\begin{itemize}
\item copies linked to a 1-vertex or 2-vertex in $G$ are colored with 3 and 4,
\item copies linked to a 3-vertex or 4-vertex in $G$ are colored with 1 and 2.
\end{itemize}
In this way we get an equitable 4-coloring of $G \circ^1 H$ with cardinalities of color classes 1 and 3 equal to $\lceil x/2 \rceil +xy$, and cardinalities of color classes 2 and 4 equal 
to $\lfloor x/2 \rfloor +xy$.

Since $G \circ ^l H = (G \circ^{l-1} H) \circ H$, we can inductively extend the equitable 4-coloring of $G \circ^{l-1} H$ into an equitable 4-coloring of $G\circ ^l H$ by coloring each of 
the uncolored copies of $H$ in $G \circ^l H$ with two colors:
\begin{itemize}
\item copies linked to a 1-vertex or 2-vertex in $G\circ ^{l-1} H$ are colored with 3 and 4,
\item copies linked to a 3-vertex or 4-vertex in $G\circ ^{l-1} H$ are colored with 1 and 2.
\end{itemize}

\item[\textnormal{\emph{Case} 2:}] $G \in Q_3$ 

Since the number of vertices in graph $G$ is even, we have to consider two subcases:
\begin{description}
\item[\textnormal{\emph{Subcase} 2.1:}] $n_G=4s$ \emph{for some} $s\geq 2$.

The thesis follows from Lemma \ref{lma4}.

\item[\textnormal{\emph{Subcase} 2.2:}] $n_G=4s+2$ \emph{for some} $s\geq 2$.

We start from an equitable 4-coloring of $G$ - this is possible due to inequality (\ref{prog}). Without loss of generality, we may assume that in this coloring the sets of 1- and 
2-vertices contain one more vertex than the sets of 3- and 4-vertices. Let $v_1$ and $v_2$ be two vertices in $G$ with colors $1$ and $2$, respectively. 
It is easy to see that $(G-\{v_1, v_2\})\circ H$ has a strong equitable $4$-coloring. 

Now, we show that $G[\{v_1, v_2\}]\circ H$ is equitably $4$-colorable, where $G[\{v_1, v_2\}]$ is the subgraph of $G$ induced by vertex set $\{v_1,v_2\}$. For $i=1,2$, let $c(v_i)=i$ and let $H_i$
be a copy of $H$ linked to $v_i$. Furthermore, for $i=1,2$ let $X_i$ and $Y_i$ be the partition sets of $H_i$. Color the vertices of $X_1$ with color $2$, vertices of $Y_1$ with color $3$, vertices of $X_2$ with color $4$, and vertices of $Y_2$ with color $1$, respectively. One can easily check that this results in an equitable $4$-coloring of $G \circ H$. 

An equitable 4-coloring of $G \circ ^l H$, in this subcase, follows from induction on $l$.

\end{description}

\item[\textnormal{\emph{Case} 3:}] $G \in Q_4$ 

The thesis follows immediately from Lemma \ref{lma4}.
\end{description}
\end{proof}

\subsection{Case $H \in Q_3$}

In this subsection we obtain some polynomially solvable cases concerning optimal equitable coloring of multicoronas $G \circ^l H$, where $H \in Q_3$. In the remaining cases we give 
1-absolute approximation algorithms.
 
\begin{theorem}
Let $G$ be any cubic graph and let $H \in Q_3$. If $G$ has a strong equitable 4-coloring, then $$\chi_=(G \circ ^l H)=4$$ for any $l \geq 1$. \label{4_cub3}
\end{theorem}
\begin{proof}
It is clear that $\chi_=(G \circ^l H) \geq 4$, for $G$ and $H$ under assumption. On the other hand, $\chi_=(G \circ^l H) \leq 4$ by Lemma \ref{lma4}, and the thesis follows. 
\end{proof}

\begin{proposition}
If $G$ is a subgraph of cubic graph on $n_G \geq 4$ vertices, where $4 | n_G$ and $H \in Q_3$, then there is an equitable $5$-coloring of $G \circ^l H$. \label{n4_cub3}
\end{proposition}
\begin{proof}
Let $n_G = 4x$ for some integer $x \geq 1$. First, let us notice that there is a strong equitable 4-coloring of $G \circ^l H$ due to inequality (\ref{prog}) and Corollary \ref{lma_cor}. 
We color equitably $G \circ^l H$ with 
4 colors in the way described in the proof of 
Lemma \ref{lma4}. In such a coloring every color is used exactly $x(n_H+1)^l$ times. Now, we have to choose vertices in each of four color classes which should be recolored to 5 so 
that the resulting 5-coloring of $G \circ^l H$ is equitable. 
It turns out that we can choose a proper number of $i$-vertices, $i=1,2,3,$ and 4, that should be recolored to 5 from partitation sets $X_1$ of $H(X_1,X_2,X_3)$ linked to vertices of 
$G \circ ^{l-1} H$ during 
creating $l$-corona product $G \circ ^l H$ from $G \circ ^{l-1} H$. Moreover, we need only copies of $H$ from this $l$-th step that were linked to vertices of $G$. Since $n_G=4x$, 
we have exactly $x$ $i$-vertices in $G$, $i \leq 4$. We will see that we need at most $x|X_1|$ $i$-vertices that should be recolored to 5. We choose them from $X_1$'s linked 
to $(i-1)$-vertices in $G$, $i=1,2,3,4$ (we use color 4 instead of color 0). To prove this, let us consider three cases.
\begin{description}
\item[\textnormal{\emph{Case} 1:}] $H(X_1,X_2,X_3) \in Q_3(t+1,t,t)$ \emph{for some odd} $t \geq 3$.

In 4-coloring of $G \circ^l H$ each of four colors is used $x(3t+2)$ times, while in every equitable 5-coloring of the corona, each of five colors must be used $\lceil (12xt+8x)/5\rceil=
2xt+x+\lceil (2xt+3x)/5\rceil$ or $2xt+x+\lfloor (2xt+3x)/5\rfloor$ times. This means that the number of vertices that should be recolored to 5 in each of the four color classes is equal 
to at most $$3xt+2x-2xt-x-\lfloor (2xt+3x)/5\rfloor = x(t+1) - \lfloor (2xt+3x)/5\rfloor <x(t+1) =x|X_1|.$$

\item[\textnormal{\emph{Case} 2:}] $H(X_1,X_2,X_3) \in Q_3(t+1,t+1,t)$ \emph{for some even} $t \geq 2$.

In 4-coloring of $G \circ H$ each of four colors is used $x(3t+3)$ times, while in every equitable 5-coloring of the corona each of five colors must be used $\lceil (12xt+12x)/5\rceil=2xt+2x+
\lceil (2xt+2x)/5\rceil$ or $2xt+2x+\lfloor (2xt+2x)/5\rfloor$ times. This means that the number of vertices that should be recolored to 5 in each of the four color classes is equal to 
at most $$3xt+3x-2xt-2x-\lfloor (2xt+2x)/5\rfloor = x(t+1) - \lfloor (2xt+2x)/5\rfloor <x(t+1) =x|X_1|.$$

\item[\textnormal{\emph{Case} 3:}] $H(X_1,X_2,X_3) \in Q_3(t,t,t)$ \emph{for some even} $t \geq 2$.

In 4-coloring of $G \circ H$ each of four colors is used $x(3t+1)$ times, while in every equitable 5-coloring of the corona each of five colors must be used $\lceil (12xt+4x)/5\rceil=2xt+
\lceil (2xt+4x)/5\rceil$ or $2xt+\lfloor (2xt+4x)/5\rfloor$ times. This means that the number of vertices that should be recolored to 5 in each of the four color classes is equal to at 
most $$3xt+x-2xt-\lfloor (2xt+4x)/5\rfloor =x(t+1) - \lfloor (2xt+4x)/5\rfloor \leq xt =x|X_1|.$$
\end{description}

This completes the proof.
\end{proof}

\begin{theorem}
If $G$ is a cubic graph on $n_G \geq 8$ vertices and $H \in Q_3$, then $$\chi_=(G \circ ^l H) \leq 5.$$ \label{5_cub3}
\end{theorem}
\begin{proof}
If $5 | n_G$, then $G$ has a strong equitable 5-coloring (due to inequality (\ref{prog})) and the thesis follows from Lemma \ref{lma} for $k=5$. We need to consider the cases where 
$n_G \bmod 5 \neq 0$. 

\begin{description}
\item[\textnormal{\emph{Case} 1:}] $n_G \bmod 5 = 1$ \emph{and} $n_G \geq 16$.

We start from a semi-equitable 5-coloring of cubic graph $G$ of type $[(n+4)/5,(n+4)/5,(n+4)/5,(n+4)/5,(n-16)/5]$ - this is possible due to Theorem \ref{semi} for $k=5$. Next, we choose 
four 1-vertices, four 2-vertices,
four 3-vertices, and four 4-vertices from the center graph $G$. They form a set $V^{16}(G)$. We consider the subgraph of $G$ induced by this vertex set - subcubic graph $G[V^{16}]$, and 
corona graph 
$G[V^{16}]\circ^l H$ being subgraph of $G \circ^l H$. $G[V^{16}]\circ^l H$ has an equitable 5-coloring due to Proposition \ref{n4_cub3}. Note, that this equitable 5-coloring of 
$l$-corona, described in the proof of
Proposition \ref{n4_cub3}, starts from a strong equitable 4-coloring of the center graph. Thus, there is possible to extend the strong 4-coloring of $G[V^{16}]$ into equitable 5-coloring of 
$G[V^{16}]\circ^l H$. Next, we consider subgraph $G[V\backslash V^{16}]$, strong equitably 5-colored, as a center graph of $l$-corona $G[V\backslash V^{16}] \circ^l H$. Due to Lemma \ref{lma} for $k=5$, 
$G[V\backslash V^{16}] \circ^l H$ has an equitable 5-coloring and this coloring is strong equitable. Furthermore, also this coloring is based on strong equitable 5-coloring of 
$G[V\backslash V^{16}]$. This means, that equitable 5 colorings of $G[V^{16}]\circ^l H$ and $G[V\backslash V^{16}] \circ^l H$ may be combined into one proper equitable 5-coloring 
of $G \circ ^l H$.

\item[\textnormal{\emph{Case} 2:}] $n_G \bmod 5 = 2$.

The idea is similar to that presented in the previous case. This time we start from a semi-equitable 5-coloring of $G$ of type $[(n+3)/5,(n+3)/5,(n+3)/5,(n+3)/5,(n-12)/5]$ - 
this is possible due to Theorem \ref{semi}.

Analogously, we choose three 1-vertices, three 2-vertices, three 3-vertices, and three 4-vertices from the graph $G$. They form a set $V^{12}(G)$. First we extend the coloring of $G$ into 
$G[V^{12}] \circ^l H$, and then into $G[V \backslash V^{12}] \circ^l H$. Finally, we obtain an equitable 5-coloring of $G \circ^l H$.

\item[\textnormal{\emph{Case} 3:}] $n_G \bmod 5 = 3$. 

This time we start from a semi-equitable 5-coloring of $G$ of type $[(n+2)/5,(n+2)/5,(n+2)/5,(n+2)/5,(n-8)/5]$ (possible due to Theorem \ref{semi}). Next, we choose, anoulogously 
to the previous case, 8 vertices of $G$, forming set $V^8(G)$. We extend the coloring of $G$ into 
$G[V^{8}] \circ^l H$, and then into $G[V \backslash V^{8}] \circ^l H$. Finally, we obtain an equitable 5-coloring of $G \circ^l H$.

\item[\textnormal{\emph{Case} 4:}] $n_G \bmod 5 = 4$.

In the last case we start from a semi-equitable coloring of $G$ of type $[(n+1)/5,(n+1)/5,(n+1)/5,(n+1)/5,(n-4)/5]$. We choose one vertex of each color $i$, $1 \leq i \leq 4$, from graph $G$. 
The vertices form the set $V^4$. First, we extend the coloring of $G[V^4]$ into an equitable 5-coloring of $G[V^4] \circ ^ l H$, in the way described in the proof of Proposition \ref{n4_cub3}. 
Next, we extend strong equitable 5-coloring of $G[V \backslash V^4]$ into strong equitable 5-coloring of $G[V \backslash V^4] \circ^l H$ (due to method described in the proof of Lemma 
\ref{lma}). Finally, we obtain an equitable 5-coloring of 
$G \circ ^l H$.
\end{description}
\end{proof}

\subsection{Case $H=K_4$}
\begin {proposition} [\cite{hf}]
If $G$ is a graph with $\chi\left(G\right)\leq m+1$, then $\chi_{=}(G \circ K_m)= m+1$. \label{complete}
\end{proposition}

\begin{theorem}
If $G$ is cubic and $l \geq 1$, then $$\chi_=(G\circ^l K_4)=5.$$ \label{k4}
\end{theorem}
\begin{proof}
Since the following inequalities hold for every cubic graph $G$:
$$ 2 \leq \chi(G) \leq 4,$$ due to Brooks theorem \cite{brooks},
so cubic graph $G$ fulfills the assumption of Proposition \ref{complete} for $m=4$ and we have $\chi_=(G \circ K_4)=5$. As $G \circ ^l K_4 = (G \circ^{l-1} K_4) \circ K_4$ and $\chi_=(G \circ ^2 K_4)=\chi_=(G \circ ^3 K_4)=\cdots = \chi_=(G \circ ^{l-1} K_4)=5$, we get immediately the thesis.
\end{proof}

\section{Conclusion}\label{sum}

In the paper we have given some results concerning the equitable coloring of $l$-corona products $G \circ^l H$, where $G$ and $H$ are cubic graphs. The main of our results are summarized in 
Table \ref{tabela1}. In the table the entry '$3$ or $4$' means that we have identified all the cases for which $\chi_=(G \circ^l H)=3$ and/or $\chi_=(G \circ^l H)=4$. The entry '$\leq 5$' 
means
merely that $\chi_=(G \circ^l H) \leq 5$.

\begin{table}[htb]
\begin{center}
\begin{tabular}{|c|*{3}{c|}}\hline

\backslashbox[20mm]{$G$}{$H$} & $Q_2$ & $Q_3$ & $Q_4$\\ \hline
\multirow{2}{*}{$Q_2(t)$} & \multirow{2}{*}{3 or 4 \scriptsize{[Thm. \ref{3col}, \ref{H_2leq4}]}} & 4 for $t$ even \scriptsize{[Thm. \ref{4_cub3}]} &\multirow{2}{*}{5 \scriptsize{[Thm. \ref{k4}]}} \\ 
& & $\leq 5$ for $t$ odd \scriptsize{[Thm. \ref{5_cub3}]} & \\\hline
$Q_3$ & 3 or 4 \scriptsize{[Thm. \ref{3col}, \ref{H_2leq4}]} & $\leq 5^*$ \scriptsize{[Thm. \ref{5_cub3}]}  &5 \scriptsize{[Thm. \ref{k4}]} \\ \hline
$Q_4$ & 4 \scriptsize{[Thm. \ref{H_2leq4}]} & 4 \scriptsize{[Thm. \ref{4_cub3}]}& 5 \scriptsize{[Thm. \ref{k4}]} \\\hline
\end{tabular}

\vspace{0.5cm}
\small{*: we remind to the reader the case, where $n_G=6$. One should check that the bound holds also for such center graphs $G$ (there are only two cubic graphs on 6 vertices).}

\vspace{3mm}
%\hspace{0.5cm} $^1:$ we have determined all the cases when $\chi_==3$ or $\chi_==4$,\\
\caption{Possible values of $\chi_=(G \circ ^l H)$ for cubical multicoronas.}\label{tabela1}
\end{center}
\end{table}

Note that our results confirm the Equitable Coloring Conjecture for graphs under consideration. This conjecture was posed by Meyer \cite{meyer} in 1973.

What about the complexity of equitable coloring of cubical multicoronas? From \cite{harder} we know that this problem is NP-hard for coronas $G \circ^l H$, $l=1$. We remain as an open question 
whether this result can be extended to arbitrary cubical coronas $G \circ ^l H$, $l \geq 2$. 

We know that ordinary coloring of cubical multicoronas can be determined in polynomial time. The 
exact values
of ordinary chromatic number of $l$-corona products under consideration are given in Table \ref{tabela2}. The appropriate coloring of $G \circ H$ is obtained by coloring $G$ with 
$\chi(G)$ colors and extending this coloring into copies of $r$-partite cubic graph $H$ linked to $i$-vertex of $G$ by coloring $r$ partition sets with $(i+1) \bmod \chi(G\circ^1 H), 
\ldots, (i+r) \bmod \chi(G\circ ^1 H)$, respectively (we use color $\chi(G \circ^1 H)$ instead of color 0). Such a coloring can be extended into copies of $H$ for bigger $l$ in the 
similar way. 

\begin{table}[htb]
\begin{center}
\begin{tabular}{|c|*{3}{c|}}\hline
\backslashbox[20mm]{$G$}{$H$} & $Q_2$ & $Q_3$ & $Q_4$\\ \hline
$Q_2$ & 3 & 4 & 5 \\ \hline
$Q_3$ & 3 & 4 & 5 \\ \hline
$Q_4$ & 4 & 4 & 5 \\ \hline
\end{tabular}
\vspace{3mm}
\caption{The exact values of $\chi(G \circ ^l H)$ for cubical multicoronas.}\label{tabela2}
\end{center}
\end{table}

Simple comparison of Tables \ref{tabela1} and \ref{tabela2} leads us to the conclusion
that our results miss the exact values by at most one color.

\end{document}